\documentclass[reqno]{amsart}

\usepackage{amsmath,amssymb}
\usepackage{mathtools}
\usepackage{graphics,epsfig,xcolor}
\usepackage[mathscr]{euscript}
\newtheorem{theorem}{Theorem}[section]
\newtheorem{proposition}[theorem]{Proposition}
\newtheorem{lemma}[theorem]{Lemma}

\theoremstyle{definition}
\newtheorem{definition}[theorem]{Definition}

\newtheorem{remark}[theorem]{Remark}
\numberwithin{equation}{section}
\numberwithin{theorem}{section}
\renewcommand{\epsilon}{\varepsilon}

\newcommand{\mc}[1]{{\mathcal #1}}
\newcommand{\bb}[1]{{\mathbb #1}}
\newcommand{\ms}[1]{{\mathscr #1}}
\newcommand{\bs}[1]{{\boldsymbol #1}}
\newcommand{\e}{\epsilon}

\newcommand{\Ent}{\mathop{\rm Ent}\nolimits}
\newcommand{\E}{\mathop{\rm E}\nolimits}
\newcommand{\He}{\mathop{\rm H_e}\nolimits}

\newcommand{\de}{\mathop{}\!\mathrm{d}}
\newcommand{\carlo}{\color{black}}

\title{Variational derivation of the homogeneous Boltzmann equation}

\author[G.\ Basile]{Giada Basile}
\address{Giada Basile \hfill\break \indent
   Dipartimento di Matematica, Universit\`a di Roma `La Sapienza'
   \hfill\break \indent
   P.le Aldo Moro 2, 00185 Roma, Italy}
 \email{basile@mat.uniroma1.it}
 \author[D.\ Benedetto]{Dario Benedetto}
 \address{Dario Benedetto \hfill\break \indent
   Dipartimento di Matematica, Universit\`a di Roma `La Sapienza'
   \hfill\break \indent
   P.le Aldo Moro 2, 00185 Roma, Italy}
 \email{benedetto@mat.uniroma1.it}
 \author[C.\ Orrieri]{Carlo Orrieri}
 \address{Carlo Orrieri \hfill\break \indent
   Dipartimento di Matematica, Universit\`a di Pavia
   \hfill\break \indent
   Via Ferrata 1, 27100 Pavia, Italy}
 \email{carlo.orrieri@unipv.it}

\begin{document}
 \begin{abstract}
   We introduce a variational formulation of the homogeneous
   Boltzmann equation, with hard-sphere cross section, which selects
   the unique energy conserving solution.
   We prove that this solution arises from the microscopic dynamics, namely Kac's walk,
   and we establish the propagation in time of entropic chaoticity,
   under the minimal assumption
   that the initial distribution is entropically chaotic.
 \end{abstract}

\maketitle
\thispagestyle{empty}

\section{Introduction}
Kinetic equations can be naturally formulated in terms of an entropy dissipation inequality.
For the homogeneous Boltzmann equation (HBE)
this formulation has been firstly established by Erbar \cite{Er2} in terms of a gradient flow with respect to a suitable metric.
One of the advantage of this approach is the possibility to derive the solution to the HBE from the underlying stochastic microscopic
dynamics, the Kac walk, passing to the limit in the entropy dissipation inequality for the particle system.
Under some moment assumptions on the initial
distribution, this has been proved in \cite{Er2}, together with the propagation of entropic chaoticity, i.e. the convergence of
the microscopic entropy  per particle to the  macroscopic one \cite{CC}.
The propagation of entropic chaoticity is a key point of the Kac program and it has been firstly established by Mischler,
Mouhot \cite{MM}, under stronger assumptions on the initial distribution, without exploiting the variational structure.

In this work we propose a novel variational characterization of HBE, written
in terms of a measure-flux pair, and related to the large deviation rate function
of the underlying microscopic dynamics.
Variational formulations involving the large deviation rate function
have been already introduced in \cite{ADPZ,MPR}.
This kind of structure can be also recognized
for the Fokker-Planck equation of continuous time Markov chains \cite{Mi,BBB}.
A different variational formulation can be found in 
\cite{Ma, Er}.
Recently in \cite{BB} the non-reversible case has been approached.

We use the variational 
formulation to derive the kinetic limit of the empirical measure of the Kac walk with hard-sphere cross section
under the minimal assumption of entropic chaoticity.
Under the same minimal condition we prove the
propagation in time of the entropic chaoticity.
We stress that 
the HBE with hard-sphere cross section has a unique energy preserving solution
\cite{MW}, but other solutions with increasing energy can be constructed
as firstly proved by Lu and Wennberg \cite{LuW}. For all these solutions,
the balance of the kinetic entropy  holds \cite{Lu}.
Therefore, a variational formulation related to the entropy  balance
cannot select the ``correct'' solution, unless  the initial datum
has some finite higher moments.
In our approach the uniqueness of the variational solution
is obtained by constraining the energy not to exceed the initial one.
This naturally follows from the microscopic derivation
from the Kac walk, under the assumption of entropic chaoticity of the initial
datum. Notice that Lu and Wennberg solutions
can be derived from the microscopic dynamics,
as shown in \cite{He}, but the 
initial data are not entropically chaotic.

In our proof
we do not require 
finite higher moments, as usually assumed in the literature, see   e.g. \cite{Sn}.
A key observation is 
that any limit point of the empirical path solves the Boltzmann equation.
The initial entropic chaoticity forces the limit points to have the initial right energy,
i.e. the one prescribed by the microcanonical initial measure.
This fact follows from a result on large deviation on the initial measure (see \cite{BBBC2,Nam}),
together with the relation between the entropy per particle and the large deviation rate function
established in \cite{Mar}.
Finally, the
uniqueness of the limit point follows from 
the uniqueness  of the energy preserving solution.

To derive the macroscopic entropy dissipation inequality from the microscopic one we exploit the Erbar's construction in \cite{Er2}, which includes both the pushforward of the empirical measure and of a measure recording collision details.
We remark that in the Erbar proof a superposition argument is needed.
We avoid it since 
in our large deviation approach, 
the flux, an observable
recording the details of the collisions, is already  included
as a  microscopic variable. Moreover, we consider the case of the  hard-sphare cross section,
which is not included in the Erbar results, due to a bound from below of the kernel.

We conclude by remarking that 
a challenging program
would be a variational characterization of the solutions to the space non-homogeneous
Boltzmann equation.
A result in this direction can be found in \cite{EH}, where  a generalization of \cite{Er2} is established for a ``fuzzy'' Boltzmann equation.

The remainder of the paper is organized as follows. Section 2 introduces the measure-flux framework and develops the variational formulation of the HBE. Section 3 connects the variational solution to the Kac walk, proving convergence and
propagation of the entropic chaoticity in time.

\section{Variational solution to the homogeneous Boltzmann equation}
We start by introducing the main objects that we will use in this paper, in an abstract setting.
Let  $\mc X$ be a Polish space 
space  and $\mc B_{\mc X}$ the
Borel $\sigma$-algebra on $\mc X$. We  denote by $\mc P(\mc X)$ the space of probability measures on $(\mc X,\mc B_{\mc X})$. 
Given $\mu$, $\nu \in \mc P(\mc X)$, we denote by $\Ent(\mu \vert \nu)$ the relative entropy of $\mu$ with respect to $\nu$, namely
\begin{equation}\label{def:Ent}
  \Ent(\mu \vert \nu)= \begin{cases}
    \int \de \mu\ln \frac {\de \mu}{\de {\nu}} & \mbox{if } \mu\ll \nu\\
    +\infty & \mbox{otherwise}.
    \end{cases}
\end{equation}
Equivalently, it can be defined using the variational formulation
\begin{equation}\label{vdef:Ent}
\Ent(\mu\vert\nu)=\sup_{\phi\in C_b(\mc X)}\Big\{\mu(\phi) - \log \nu(e^\phi) \Big\}.
\end{equation}  
The definition of the relative entropy can be extended to finite positive measures $\mc V$, $\tilde{\mc V}$  on
$\mc X$, by considering the functional
\begin{equation}\label{vdef:E}
\E(\mc V\vert \tilde{\mc V})=\sup_{\phi\in C_b(\mc X)}\big\{\mc V(\phi)- \tilde{\mc V}(e^\phi -1)   \big\},
\end{equation}  
which is equivalent to
\begin{equation}\label{def:E}
\E(\mc V\vert \tilde{\mc V})=\begin{cases}
    \int \de \mc V\ln \frac {\de \mc V}{\de \tilde{\mc V}} - {\carlo \mc V(\mc X) + \tilde{\mc V}(\mc X)} & \mbox{if } \mc V\ll \tilde{\mc V}\\
    +\infty & \mbox{otherwise}.
    \end{cases}
\end{equation}  
By the variational definition, both $\Ent$ and $\E$ are non-negative, convex and lower semicontinuous and they vanish if and only if the two measures are the same.
For every Polish space $Y$ and Borel measurable function $\phi: X \to Y$ it holds that
\begin{equation}\label{AGS}
  \E(\mc V\vert \tilde{\mc V}) \ge E(\phi_{\#} \mc V \vert \phi_{\#}  \tilde{\mc V}),
\end{equation}
see \cite[Lemma 9.4.5]{AGS}.

\subsection*{The homogeneous Boltzmann equation} We consider the homogeneous Boltzmann equation (HBE) with hard-sphere kernel,  in the weak form. Fix $T>0$ and $d\geq 2$. Let $C\big([0,T]; \mc P (\bb R^d)\big)$ be the
set of continuous paths on $\mc P(\bb R^d)$ endowed with
the topology of uniform convergence. A probability measure
$P\in C([0,T], {\carlo\mc P} (\bb R^d))$ is a weak solution to the homogeneous Boltzmann equation if
has finite first moment and satisfies
\begin{equation}\label{wHBE}\begin{split}
  &  P_T(\phi_T)-P_0(\phi_0)-\int_0^T\!\!\de t\,P_t(\partial_t \phi)\\
  &  = \frac 1 2 \int_0^T \!\!\de t\int_{\bb R^{2d}}\int_{S^{d-1}}\de \omega \,P_t(\de v) P_t(\de v_*)B(v-v_*,\omega)
    \overline\nabla \phi_t(v,v_*, v',v_*')
\end{split}\end{equation}  
for any $\phi\in C_b([0,T], \bb R^d)$ with continuous derivative in $t$.
Here $\overline
\nabla\phi(v,v_*, v',v_*')= \phi(v')+\phi(v'_*)-\phi(v)-\phi(v_*)$, with $v', v'_*$ outgoing velocities given by the  collision rule
$$
v'=v-\big((v-v_*)\cdot\omega\big)\omega,\qquad v'_*=v+\big((v-v_*)\cdot\omega\big)\omega.
$$
The hard-sphere collision kernel $B$ has the form
\begin{equation}\label{def:B}
B(v-v_*,\omega) = \frac 1 2 |(v-v_*)\cdot\omega|.
\end{equation}

In the same spirit of \cite{BBB, BBBC2} we are going to rewrite the homogeneous Boltzmann equation in terms of a balance
equation and a constitutive equation. To this aim we consider the flow as a dynamical variable. 
We denote by $\ms M$ the  subset of the positive finite  measures $Q$  on
$[0,T]\times \bb R^{2d}\times \bb R^{2d}$  that satisfy
$Q(\de t; \de v,\de v_*, \de v',\de v_*')=Q(\de t; \de v_*,\de v, \de v',\de v_*') =Q(\de t; \de v,\de v_*, \de v'_*,\de v')$.
We consider $\ms M$ endowed with the weak topology and
the corresponding Borel $\sigma$-algebra.
By definition, the weak topology is 
the weakest topology such that the map
$Q \mapsto Q(F)$ is continuous for each $F$ in
$C_{\mathrm b}([0,T]\times \bb R^{2d}\times \bb R^{2d})$.
The  product space
$C\big([0,T]; \mc P (\bb R^d)\big)\times \ms M$ is
endowed with the product topology, and 
we  denote by $\ms S$ the subspace of
$C\big([0,T]; \mc P (\bb R^d)\big)\times \ms M$
given by the pairs
$(P,Q)$ that satisfies the balance equation
\begin{equation}
  \label{bal}
  P_T(\phi_T)-P_0(\phi_0)-\int_0^T\! \de t\, 
  P_t(\partial_t \phi_t) = Q(\overline \nabla \phi)
 \end{equation}
for each $\phi\in C_{\rm{b}}([0,T]\times \bb R^d)$.

\begin{definition}[Measure-flux solution to the HBE]
    \label{def:mfboltz}

    We say that a measure-flux pair $(P,Q)\in \ms S$
    is a  solution to the homogeneous Boltzmann equation
    if and only if $Q= Q^{P\otimes P}$, where
    \begin{equation*}
      \begin{split}
      &Q^{P\otimes P}(\de t,\de v, \de v_*,\de v',\de v'_*) \\
    &\coloneqq \de t \, \frac 12 P_t(\de v) P_t(\de v_*)
    \int_{S^{d-1}} \de \omega 
    B(v - v_*,\omega)
  \delta_{v-((v-v_*)\cdot\omega)\omega}(\de v')
      \delta_{v+((v-v_*)\cdot\omega)\omega}(\de v'_*).
      \end{split}
 \end{equation*}
\end{definition}
The above definition is justified by the fact that
$P$ with finite first moment solves \eqref{wHBE} if and only if  $(P,Q^{P\otimes P})\in \mc S$.

\subsection*{Variational formulation of the HBE} We provide a formulation of the homogeneous Boltzmann equation in terms of  an entropy dissipation inequality, in the same spirit of \cite{BBB} (see also \cite{BBBO}).

Let $\ell$ be the Lebesgue measure on $\bb R^d$.
For any  probability measure $P$ with finite first moment and which is absolutely continuous with respect to $\ell$, 
we set 
\begin{equation}\label{def:H}
  \mc H(P)=\int \de P \log f
  \end{equation}
where $\de P=f\de \ell$.
Note that the assumption on the finiteness of the first momement guarantees that the negative part
of the integral is finite.

Let  $\bs\zeta\colon \bb R^d\to [0, +\infty)\times \bb R^d$ be the map $\bs\zeta(v)=(\zeta_0, \zeta)(v)=(|v|^2/2, v)$. 
For fixed $e>0$, we denote by
$\ms P_e$ the subset of $\mc P(\bb R^d)$ given by the
probabilities with vanishing mean and second moment bounded by $2e$, i.e.
$P(\zeta_0)\le e$ and $P(\zeta)=0$. Observe that  $\ms P_e$
is a compact convex subset of $\mc P(\bb R^d)$, and we equip it with
the relative topology and the corresponding Borel $\sigma$-algebra.
We denote by $\ms S_e$ the subset of $\ms S$  of the pairs $(P,Q)$ with $P\in C([0,T],\ms P_e)$.

We denote by  $\Upsilon\colon [0,T]\times \bb R^{2d}\times \bb R^{2d}\to [0,T]\times \bb R^{2d}\times \bb R^{2d}$ the map
 that exchanges the incoming and outgoing velocities, i.e.
\begin{equation}
  \label{ups}
  \Upsilon (t,v,v_*,v',v'_*) = (t,v',v_*',v,v_*).
\end{equation}

\begin{definition}[Variational solution to the HBE]
  \label{def:var}
  Consider  $(P,Q)\in {\ms S}$, such that ${\carlo e_0 :}= P_0(\zeta_0)<+\infty$ and $\mc H(P_0)<+\infty$.
  We say that $(P,Q)\in \ms S$ is a {\carlo variational} solution to the HBE \eqref{wHBE}
  if and only if $(P,Q)\in \ms S_{e_0}$ and
\begin{equation}
  \label{gfB}
  \mc H(P_T) + 
    \E(Q \vert Q^{P\otimes P}) +
    \E(Q \vert \Upsilon_{\#}Q^{P\otimes P}) \,\le \, \mc H(P_0).
  \end{equation}
\end{definition}

As we will show  in Proposition \ref{prop:QQpi},  the reverse inequality 
 holds for any $(P,Q)\in \ms S_e$, $e>0$,  such that $\mc H(P_0)<+\infty$.
Observe that Definition \ref{def:var} 
selects the solutions
whose energy does not exceed its initial value.

\begin{proposition}\label{prop:QQpi}
  Fix $e>0$. For any $(P,Q)\in {\ms S}_{e}$, such that $\mc H(P_0)<+\infty$
\begin{equation}
  \label{revineq}
  \mc H(P_T) + 
    \E(Q \vert Q^{P\otimes P}) +
    \E(Q \vert \Upsilon_{\#}Q^{P\otimes P}) \,\ge \, \mc H(P_0).
\end{equation}
Moreover, equality holds iff $Q=Q^{P\otimes P}$, i.e.
$(P,Q)$ is 
a measure-flux solution to the HBE.
\end{proposition}

\begin{proof}
  We restrict to the pairs $(P,Q)\in \ms{S}_e$ such that the functionals on the {\carlo left} hand side of \eqref{revineq} are finite, otherwise the inequality is trivial.
  Denote by $M_e$ the $d$-dimensional Maxwellian of  zero mean and energy $e$, namely with inverse temperature $\beta=\frac d {2 e}$. For any $\pi\in \ms P(\bb R^d)$, define the functional
  \begin{equation}\label{def:He}
    \He(\pi)=
    \begin{cases}\Ent(\pi\vert M_e)+\beta \big(e-\pi(\zeta_0)\big) & \mbox{if } \pi\in\ms P_e\\
      +\infty & \mbox{otherwise}
      \end{cases}
  \end{equation}  
The following equality holds (\emph{entropy balance})
  \begin{equation}
    \label{cr}
     \He(P_0) +\E(Q \vert Q^{P\otimes P})= \He(P_T) + \E(Q \vert \Upsilon_{\#}Q^{P\otimes P}),
   \end{equation}
  as stated in \cite{BBBH}, Proposition 3.1.
  The identity  holds in the sense that if either side is finite,
then also the other one is finite and equality holds.
  By adding to both sides $\E(Q \vert Q^{P\otimes P})$, {\carlo which is finite by assumption, } we obtain 
\begin{equation*}
 \He(P_T) + \E(Q \vert \Upsilon_{\#}Q^{P\otimes P})+\E(Q \vert Q^{P\otimes P})=\He(P_0) +2 \E(Q \vert Q^{P\otimes P})\geq \He(P_0),
\end{equation*}
where equality holds if and only if {\carlo $\E(Q \vert Q^{P\otimes P}) = 0$ so that}  $Q=Q^{P\otimes P}$. The proof is concluded by observing that  $\Ent(\pi\vert M_e)=\mc H(\pi)+\frac d {2e}\pi(\zeta_0)+\frac d 2 \ln \frac{4\pi e}d$, therefore by definition \eqref{def:He}, for any $\pi \in \ms P_e$
\begin{equation}
  \label{heh}
\He(\pi)=\mc H (\pi)
+\frac d 2\Big(\ln\frac{4\pi e}d +1 \Big).
\end{equation}  
\end{proof}

\begin{remark}
  Observe that \eqref{cr} holds for any measure-flux solution to the HBE,
  including the Lu-Wennberg solutions, which have increasing energy, as already stated in \cite{Lu}.
  We stress that the required condition $(P,Q)\in \ms S_{e_0}$ in Definition \eqref{def:var},
  which exclude Lu-Wennberg solutions, naturally follows from the microscopic derivation of \eqref{gfB}
  from the Kac walk, under the assumption of entropic chaoticity of the initial datum,  as we will explain in detail.
    
\end{remark}

\begin{theorem}\label{theo:eu}
  There exists a unique variational solution $(P,Q)$ to the homogeneous Boltzmann equation.
  Moreover $P_t(\zeta_0)=P_0(\zeta_0)$ for
  all $t\in [0,T]$.
\end{theorem}
\begin{proof}
The existence follows from the derivation from the Kac walk,
stated in Theorem \ref{main_theo} and Remark \ref{remark_encaos}.
  By Proposition \ref{prop:QQpi}, equality in \eqref{gfB}
  holds iff $Q=Q^{P\otimes P}$, i.e. $(P,Q)$ is a weak solution
  to the HBE. Since for any weak solution to the HBE the energy
  is not decreasing \cite{Lu, MW}, and by Definition \ref{def:var}  the energy does not exceed its initial value, then it is conserved.
  The statement follows by the uniqueness of the energy preserving solution to
  the HBE, see \cite{MW}. 
\end{proof}
\begin{remark}
  Existence of energy preserving solutions to the homogeneous Boltzmann equation under the only assumption of finite second moment
  can be found in  \cite[Theorem 1.1]{MW}.
  This guarantees the existence of variational solutions independently on the microscopic derivation.
\end{remark}

\subsection{Relation with formulation in \cite{BBB, BBBO}}
In \cite{BBB, BBBO} the authors proposed a variational formulation of the linear Boltzmann equation and of the homogeneous Boltzmann-type equation respectively, involving  dual dissipation potential with the $\cosh$ structure. More recently, in \cite{EH} a variational formulation of this kind has been proposed for a fuzzy Boltzmann equation. In this subsection  we enlighten the relation with the variational formulation given in Definition \ref{def:var}. We will omit all the technical details.

We start by introducing a pair of functionals $\mc D$ and $\mc R$.
Fix $e>0$. Consider $\pi\in \ms P_e$  such that $\Ent(\pi\vert M_e)<+\infty$ and denote $f$ the density of $\pi$ w.r.t. the Lebesgue measure $\ell$. 
 We define
  \begin{equation}
    \label{eq:Dexpl}
    \begin{split}
      \mc D(\pi) =&
      \int \de v \de v_* \de \omega f(v) f(v_*) B(v-v_*,\omega) 
      \\&- \int \de v \de v_* \de \omega
      \sqrt{
        f(v) f(v_*)f_t(v') f(v'_*)}
      B(v-v_*,\omega)
    \end{split}
  \end{equation}
  By direct computation, $\mc D$ is one half of the Dirichlet form
  for continous time Markov chains computed on $\sqrt{ff_*}$, with the usual notation $f_*=f(v_*)$ (see also \cite[Lemma 6.2]{BBBO}).

For any $(P,Q)\in \ms S_e$, define the (finite)  measure $R$ on $[0,T]\times \bb R^{2d}\times \bb R^{2d}$ as
\begin{equation*}
  \de R^{P\otimes P}=\de t\frac 1 2 \sqrt{
    f_t(v) f_t(v_*)f_t(v') f_t(v'_*)}
  B(v-v_*,\omega)\de\omega\de v \de v_*,
\end{equation*}
where $P_t(\de v)=f_t(v)\de v$. The functional $\E(Q\vert R^{P\otimes P})$ corresponds to the ``kinematic cost'' defined in \cite{BBB}.
By direct computation 
  \begin{equation}
    \label{eq:drjj}
    \int_0^T \mc D(P_t) +
    \E(Q\vert R^{P\otimes P}) = \E(Q \vert \Upsilon_{\#}Q^{P\otimes P})+\E(Q \vert Q^{P\otimes P}).
  \end{equation}

  \section{Convergence from the microscopic dynamics}
  In this section we prove that the empirical measure of the Kac walk converges  to the unique {\carlo variational} solution to the HBE from the Kac model. As a by-product, we prove entropic propagation of chaos under the only assumption that the initial particle distribution is entropically chaotic.

\subsection{The Kac walk}
We consider the Kac walk given by the Markov process $(\bs v(t))_{t\geq 0}$ on the configuration space $\big(\bb R^d\big)^N$, with  $d\geq 2$,
whose generator acts on bounded continuous functions
$f\colon \big(\bb R^d\big)^N\to \bb R$ as
\begin{equation*}
  \mathcal L_N f(\bs v)=\frac 1 N \sum_{\{i, j\}} L_{i,j} f (\bs v).
\end{equation*}  
Here the sum is carried over the unordered pairs
$\{i, j\}\subset \{1,.., N\}$, $i\neq j$, and
\begin{equation*}
  L_{i,j} f(\bs v) =
  \int_{\bb S_{d-1}}\!\! \de \omega \,B(v_i- v_j, \omega)\big[f
  \big(T^{\omega}_{i,j} \bs v\big ) -f(\bs v)   \big]. 
\end{equation*}  
Here $S^{d-1}$ is the sphere in $\bb R^d$, the post-collisional vector of velocities is given by
\begin{equation}
  \label{rules}
  \big(T^{\omega}_{i,j} \bs v\big )_k = \begin{cases}
    v_i + (\omega \cdot (v_j-v_i))\omega  & \textrm{if } k=i\\
    v_j - (\omega \cdot (v_j-v_i))\omega  & \textrm{if } k=j\\
    v_k & \textrm{otherwise},
    \end{cases}
\end{equation}  
and the collision kernel $B$ is given by
\begin{equation}\label{eq:B}
  B(v-v_*, \omega)=\frac 12 |(v-v_*)\cdot \omega|.
\end{equation}

The collisional dynamics preserves the total particle number, momentum and energy, given by the integrals of
\begin{equation}\label{def:zeta}
\bs \zeta\colon \bb R^d\mapsto [0, +\infty)\times
\bb R^d, \qquad\bs \zeta =(\zeta_0,
\zeta)(v)=(|v|^2/2, v).
\end{equation}
Therefore it  can be
restricted to the set $\Sigma^N_{e,u}$
\begin{equation*}\label{sig}
  \Sigma^N_{e,u}\coloneqq \Big\{\bs v \in\big(\bb R^d\big)^N \colon\,
  \frac 1 N\sum_{i=1}^N \bs \zeta(v_i) = (e,u)
  \Big\}. 
\end{equation*}
By Galilean invariance, we can then choose $u=0$ and set from now on $\Sigma^N_e=\Sigma^N_{e,0}$.
The Markov process $(\bs v(t))_{t\geq 0}$ is reversible with respect  to $\alpha^N$, the uniform probability measure  on $\Sigma^N_e$. Moreover, since the set $\Sigma^N_{e}$ is compact and the collision kernel $B$ is such that one can reach any open susbset 
of $\Sigma^N_{e}$ in a finite number of collisions, one can prove that the process is ergodic, as already stated
in \cite{MM}.

Fix hereafter $T>0$. Given a probability $\nu$ on $\Sigma^N_e$ we denote
by $\bb P_\nu^N$ the law of the process $(\bs v(t))_{t\geq 0}$ on the time interval $[0,T]$, starting from $\nu$.
Observe that $\bb P_\nu^N$ is a probability on the Skorokhod
space $D([0,T];\Sigma^N_e)$.
As usual, if $\nu=\delta_{\bs v}$ for some $\bs v \in \Sigma^N_e$, the
corresponding law is simply denoted by $\bb P_{\bs v}^N$.

We denote by  $P^N_t$ the law of the Markov chain at time $t$,
and by $f^N_t$ its density with respect to the invariant measure $\alpha^N$. Then $f^N$ satisfies the Kac master equation
\begin{equation}\label{KME}
\partial_t f^N_t = \mathcal L_N f^N_t.
\end{equation}

\subsection{Variational formulation of the Kac master equation}

We rewrite the master equation for the the Kac walk in term of measure-flux pair.
Given $T>0$,
let  $C\big([0,T]; \mc P (\Sigma^N_e)\big)$ be the set of continuous paths on $\mc P(\Sigma^N_e)$
endowed with the topology of uniform convergence, and  denote by 
$\mc{M}^N$ 
the set of finite measures on $[0,T]\times \Sigma^N_e\times \Sigma^N_e$
endowed with the weak topology and the corresponding Borel $\sigma$-algebra.
Given by $P^N\in C\big([0,T]; \mc P (\Sigma^N_e)\big)$, define 
the measure $\mc V^{P^N}\in \mc{M}^N$ as
\begin{equation}\label{def:VP}
  \mc V^{P^N}(\de t, \de \bs v, \de \bs u)=
  \de t\, P^N_t(\de \bs v)\frac 1 N \sum_{i< j}\int_{\bb S^{d-1}}\de \omega B(v_i-v_j,\omega) \delta_{T^\omega_{i,j}(\bs v)}(\de \bs u).
\end{equation}
Then $P^N$ is the law of the Kac walk iff the pair $(P^N,\mc V^{P^N})$ satisfies the balance equation 
\begin{equation}\label{Mbe}
  P^N_T(F_T) - P^N_0(F_0) - 
\int_0^T \!\de t \,P^N_t(\partial_tF) 
=  \int\!\mc V^{P^N}(\de t,\de \bs v, \de \bs u) \, 
\big[F_t(\bs u) -F_t(\bs v)\big],
\end{equation}
for all bounded continuous functions $F\in C_b([0,T]\times \bb R^{dN})$, with continuous and bounded derivative
in time.

With a slight abuse of notation, we still denote by $\Upsilon$  the involution operator 
$\Upsilon \colon [0,T]\times \big(\bb R^{dN}\big)^2 \to [0,T] \times\big(\bb R^{dN}\big)^2$
that exchanges the incoming and outgoing velocities, namely
\begin{equation*}
\Upsilon (t, \bs v,\bs u)=(t, \bs u,\bs v).
\end{equation*}  
For any solution $P_t$ of the Kac
master equation \eqref{Mbe} with  $\Ent (P^N_0\vert \alpha^N) < +\infty$,
the entropy production equality reads
\begin{equation}
  \label{eq:EntN}
  \Ent(P^N_T\vert \alpha^N)+\E(\mc V^{P^N} \vert {\Upsilon_{\#}}\mc V^{P^N}) = \Ent(P^N_0\vert \alpha^N).
\end{equation}

\subsection*{Empirical observables}

Recall that 
$\ms P_e$ is the subset of $\mc P(\bb R^d)$ given by the
probabilities with 
$\pi(\zeta_0)\le e$ and $\pi(\zeta)=0$.
Let $D\big([0,T]; \ms P_e\big)$ be the set of
$\ms P_e$-valued c{\'a}dl{\'a}g paths endowed with the
Skorokhod topology and the corresponding Borel $\sigma$-algebra.
The \emph{empirical measure} is the map $\pi^N \colon \Sigma^N_e \to
\ms P_e$ defined by 
\begin{equation}
  \label{1}
  {\carlo(v_1,\ldots, v_n) \mapsto } \pi^N(\bs v)\coloneqq \frac 1 N \sum_{i=1}^N \delta_{v_i}.
\end{equation}
With
a slight abuse of notation we denote also by $\pi^N$ the map from
$D\big([0,T]; \Sigma_e^N \big)$ to ${\carlo D}\big([0,T]; \ms P_e\big)$
defined by $\pi^N_t((\bs v(s))_{s \in [0,T]})\coloneqq \pi^N(\bs v(t))$, $t\in [0,T]$.

Recall that  $\ms M$ is the  subset of the finite measures $Q$  on
$[0,T]\times \bb R^{2d}\times \bb R^{2d}$  that satisfy
$Q(\de t; \de v,\de v_*, \de v',\de v_*')=Q(\de t; \de v_*,\de v, \de v',\de v_*') =Q(\de t; \de v,\de v_*, \de v'_*,\de v')$,
endowed with the weak topology and
the corresponding Borel $\sigma-$algebra.

The \emph{empirical flow} is the map $Q^N\colon D\big([0,T]; \Sigma^N_e\big) \to \ms M$
defined by
\begin{equation*}
(v_1(\cdot),\ldots, v_n(\cdot)) \mapsto Q^N(\bs v(\cdot))
\end{equation*}
such that for every $F \in C_b([0,T]\times \bb R^{2d}\times \bb R^{2d}; \bb R)$
it holds
\begin{equation}
  \label{2}
  Q^N(\bs v) (F) \coloneqq \frac 1N
  \sum_{\{i,j\}} \sum_{k\ge 1} F^{\mathrm s}\big(\tau^{i,j}_k;
  v_i(\tau^{i,j}_k-),v_j({\tau^{i,j}_k}-),
  v_i(\tau^{i,j}_k),v_j(\tau^{i,j}_k)\big)\,, 
\end{equation}
where $F^{\mathrm s}(t; v, v_*, v',v_*') = (F(t; v, v_*, v',v_*')+F(t;  v_*, v, v',v_*') + F(t;  v, v_*, v'_*,v') +
F(t;  v_*, v, v'_*,v'))/4$.
Here $(\tau^{i,j}_k)_{k\ge 1}$ are the
jump times of the pair $(v_i,v_j)$ and  $v_i(t-) = \lim_{s\uparrow t} v_i(s)$.
This definition is motivated by the physical model, for which
 $Q^N$ has to be  symmetric with respect to the exchange of incoming and outgoing particles,
respectively. 
In view of the conservation of the energy and momentum, the
measure $Q^N(\de t;\cdot)$ is supported on 
$\ms E \coloneqq \{
\bs \zeta(v)+\bs \zeta(v_*)=\bs \zeta(v')+\bs \zeta(v_*)
\}\subset \bb R^{2d}\times \bb R^{2d}$.

We extend the set $\ms S_e$ to the pairs $(P,Q)\in D([0,T]; \ms P_e\big)\times \ms M$.
For each $\bs v \in \Sigma^N_e$, with $\bb P^N_{\bs v}$ probability
one, the pair $(\pi^N,Q^N)$ belongs to $\ms S_e$.
We will denote by $\Theta^N$ the law of $(\pi^N,Q^N)$, namely $\Theta^N = (\pi^N,Q^N)_{\#} \bb P^N_{\nu}$. Therefore,
$\Theta^N$ is a probability measure on $D\big([0,T]; \ms P_e\big)\times \ms M$.
Denote by $\Xi^N$ the first marginal of $\Theta^N$, namely the law of the empirical path $\pi^N$.
For any  $t\in[0,T]$  let $e_t:D([0,T];\ms P_e)\to \ms P_e$ be the evaluation map  $e_t(P) = P_t$, and denote by  
$\Xi^N_t:=(e_t)_{\#}\Xi^N$. Note that $\Xi^N_t = \pi^N_\#P^N_t$,  which is the law of the empirical measure at time $t$.

\subsection{Convergence to the homogeneous Boltzmann equation}

For each $N$ let $P^N_0\in \mc P(\Sigma^N_e)$ be the initial
distribution of the Kac walk. We say that $P^N_0$ is $P_0$-\emph{chaotic}
if for any $k$
$$\lim_{N\to +\infty }P_0^{N, (k)}= {P_0}^{\otimes k},$$
where we denote by $P_0^{N,(k)}$ the $k$-marginal of $P^N_0$.
Note that this definition is equivalent to $\Xi^N_0 \rightharpoonup \delta_{P_0}$.

Recall that $\pi^N_\# \alpha^N$ weakly converges to $M_e$.
We say that the initial distribution $P^N_0$ is $P_0$-\emph{entropically chaotic} if it is $P_0$-chaotic and
$$
\lim_{N\to +\infty} \frac 1 N \Ent(P^N_0\vert \alpha^N)=\Ent(P_0\vert M_e).
$$

We can now formulate our main result.

\begin{theorem}\label{main_theo}

  Let $P^N_0\in \mc P(\Sigma^N_e)$ be the initial distribution of the Kac walk. Assume  that 
 $P_0^N$ is $P_0$-entropically chaotic. 
  Then $P_0(\zeta_0)=e$, and 
  $$\Theta^N\rightharpoonup \delta_{P, Q^{P\otimes P}}\ \text{ as } N\to \infty,$$
  where $(P,  Q^{P\otimes P})$ is the unique variational solution to the homogeneous Boltzmann
  equation \eqref{gfB} with initial datum $P_0$.
  
Moreover, for every $t\in [0,T]$
  \begin{equation*}
\lim_{N\to +\infty} \frac 1 N \Ent(P^N_t\vert \alpha^N)=\Ent(P_t\vert M_e).
  \end{equation*}  
\end{theorem}

\begin{remark}

  We emphasize that the entropic chaoticity forces $P_0$ to have energy $e$, as we show in the proof of Lemma \ref{lemma:convent}.

\end{remark}

\begin{remark}
\label{remark_encaos}


  Note that analogous results in the literature require an initial measure $P_0$ with
  finite fourth moment \cite{MM,Er2}.
  This condition is used in the  proof of the propagation of entropic chaoticity
  and for the construction of a $P_0-$entropically chaotic sequence.
  Our strategy does not require this condition in the propagation step.
  Furthermore, 
  we stress that this condition is not even necessary in order to 
  construct a $P_0-$entropically chaotic sequence.
  For example, if  $P_0$ has polynomial tails, with finite 
  $2+\e$ moment, 
  one can adapt the classical constuction
  by using local CLTs for
  $\alpha-$stable distributions, with $\alpha>1$.
  Moreover, as stated in \cite[Theorem 10]{CC}, a diagonal argument can be used to generalize the costruction
  to any $P_0$ with only finite energy and entropy.

  

\end{remark}

The proof follows the standard strategy of showing that  the sequence $\Theta^N$ is relatively compact and then identifying the limiting points using the variational structure.

\begin{lemma}[Limit points]\label{lemma:lp}
  The sequence of probability measure $(\Theta^N)_{N\geq 1}$ is tight.
  In particular, all the limit points $\Theta$ of $(\Theta^N)_{N\geq 1}$ concentrate on pairs $(P, Q)\in \ms S_e$
  such that $P \in C([0,T],\ms P_e)$
  and $Q=Q^{P\otimes P}$
\end{lemma}
We remark that this implies that all the limit points of the empirical measure satisfies the HBE.
\begin{proof}
  In \cite[Lemmata~4.1, 4.2]{BBBC2} it is proven  
  that the sequence $\Theta^N$ is tight, and it concentrates on pairs $(P, Q)\in \ms S_e$
  such that $P \in C([0,T],\ms P_e)$.   It remains to show that $Q=Q^{P\otimes P}$, $\Theta$-almost surely.

  Let $(\Theta^{N_k})_{k\geq 1}$ be a weakly converging subsequence of $(\Theta^N)_{N\geq 1}$ with $\Theta$ as limit point.
  By Skorokhod representation theorem, 
  we can realize all $(\pi^{N_k}, Q^{N_k})$, $k\geq 1$ on a common probability space with
  probability measure  $\Gamma$, such that $(\pi^{N_k}, Q^{N_k}) \to (P, Q)$ in $C([0,T];\ms P_e) \times \ms M$,
    $\Gamma$-almost surely, with $(P,Q)_\# \Gamma = \Theta$.
  
  For any empirical  pair $(\pi^N, Q^N)\in \ms S_e$, denote by $Q^N_{[0,t]}$ the restriction of $Q^N$ to the time
  interval $[0,t]$, with $t\in(0, T]$.
  For any bounded continuous function $F\colon ([0,T]\times \mathbb R^{2d}\times \mathbb R^{2d})\to \bb R$, the process
  \begin{equation*}
    M^{N,F}_t= Q^N_{[0,t]}(F)-\int_0^t\de s\iint_{\mathbb R^{2d}}
    \frac 1 2 \pi_s^N(\de v)\pi_s^N(\de v^*)\int_{S^{d-1}}\de\omega B(v-v_*,\omega)F(s,v, v_*, v', v'_*)
  \end{equation*}
  is a {\carlo c\`adl\`ag} martingale, see \cite{No},
  with predictable increasing quadratic variation
  \begin{equation*}
\langle M^{N,F}\rangle_t=\frac 1 N \int_0^t\de s\iint_{\mathbb R^{2d}} \frac 1 2 \pi_s^N(\de v)\pi_s^N(\de v^*)\int_{S^{d-1}}\de\omega B(v-v_*,\omega)F^2(s,v, v_*, v', v'_*).
  \end{equation*}
  Since 
  $$
  \Gamma\big(\langle M^{N,F}\rangle_T \big)\leq \frac c N  T \|F^2\|_{\infty} 
  $$
with $c=c(e)$ finite,
\begin{equation*}\begin{split}
   \Gamma\big(|M^{N,F}_t|\big)=
    \Gamma\big(\vert Q^N(F)-Q^{\pi^N\otimes \pi^N}(F)\vert \big)\leq \frac {c}{\sqrt N},
\end{split}\end{equation*}  
{\carlo possibly changing the value of the constant} $c$.

Passing to the limit along the converging subsequence, the first term converges
$\Gamma$-a.s. to $Q(F)$. The second term  converges to $Q^{P\otimes P}(F)$
thanks to the convergence of $\pi^N$ in $C([0,T];\ms P_e)$,
therefore we obtain
\begin{equation*}
\Gamma\big(\vert Q(F)-Q^{P\otimes P}(F)\vert \big)=0,
\end{equation*}
which implies  that
$\Theta\big(\vert Q(F)-Q^{P\otimes P}(F)\vert \big)=0$, so that 
{\carlo $\Theta$}-a.s. $Q=Q^{P\otimes P}$.
\end{proof}

Recall that $\Xi^N$ is the first marginal of $\Theta^N$, namely the law of the empirical path $\pi^N$, and
$\Xi^N_t:=(e_t)_{\#}\Xi^N$ is the law of the empirical measure at time $t$.
{\carlo Define $\mc V^{P^N} = \mc V^{P^N}_t \de t$ and
  $\Upsilon_\#\mc V^{P^N} =  \tilde{\mathcal V}^{P^N}_t \de t$.  
  Given a  pair of velocities $\bs v, \bs u$ in the support
  of $\mc V^{P^N}_t$, there exists a unique pair $(i,j)$ and $\omega
  \in S^{d-1}$  such that $\bs u= T^{i,j}_\omega \bs v$. 
  We push forward 
  $\mc V^{P^N}$ and $\Upsilon_{\#}\mc V^{P^N}$ by the map
  $\Phi: \Sigma_e^N \times \Sigma_e^N \to (\ms P_e)^2 \times (\bb R^d)^2
  \times (\bb R^d)^2$ given by
  $$
  \Phi(\bs v, \bs u)= \big(\pi^N(\bs v), \pi^N(\bs u), v_i, v_j,u_i ,u_j
  \big),
  $$
This defines a pair of measures 
$\beta^N:=  \Phi_{\#} \mathcal V^{P^N}$ and  
$\tilde \beta^N:= \left(\Phi \circ \Upsilon \right)_{\#}\mathcal V^{P^N}$ on 
$[0,T]\times (\ms P_e)^2 \times (\bb R^d)^2\times (\bb R^d)^2$, of the form
\begin{equation}\label{def:beta}
\begin{aligned}
  \beta^N(\de t, \de \eta, &\de \zeta, \de v, \de v_*, \de v', \de v'_*) \\
  &=N \de t \, \delta_{\eta^{N,v,v_*,v',v'_*}}(\de \zeta)
    \,
  \Xi^N_t(\de \eta) \, q^{\eta\otimes \eta}(\de v, \de v_*,
  \de v', \de v'_*)\\
  \tilde \beta^N(\de t, \de \eta, &\de \zeta, \de v, \de v_*,
\de v', \de v'_*) \\
  &= N \de t\, \delta_{\zeta^{N,v',v'_*,v',v'_*}}(\de \eta) \,
    \Xi^N_t(\de \zeta)\,  (\Upsilon_\# q^{\zeta\otimes \zeta})
    (\de v, \de v_*, \de v', \de v'_*),
\end{aligned}
\end{equation}
where, for any measure $\mu$,  $\mu^{N,v,v_*,u, u_*} = \mu
+ \frac{1}{N}(\delta_{u} + \delta_{u_*} - \delta_{v} - \delta_{v_*})$,
and
\begin{equation}\label{def:q}
  \begin{split}
  q^{\mu\otimes \mu}&(\de v, \de v_*, \de v', \de v'_*)  \coloneqq
  \frac{1}{2}  \mu(\de v)\mu(\de v_*)\\
  &\int_{S^{d-1}} \de \omega 
  B(v - v_*,\omega)
  \delta_{v-((v-v_*)\cdot\omega)\omega}(\de v')
  \delta_{v+((v-v_*)\cdot\omega)\omega}(\de v'_*).
  \end{split}
\end{equation}
Precisely, for any test function $F \in C_b( [0,T]\times (\ms P_e)^2 \times \bb R^{2d}\times \bb R^{2d}; \bb R)$,
\begin{equation*}
  \begin{split}
    \int F \, \de \beta^N  &
    = N \int_0^T \!\!\de t \int \Xi_t^N(d\eta) \, \\
  &\int q^{\eta\otimes \eta}  (\de v,
  \de v_*,\de v',\de v'_*)
      \, F\bigl(t, \eta, \eta^{N,v,v_*,v',v_*}, v, v_*, v',v'_*\bigr),
      \end{split}
    \end{equation*}
    and 
\begin{equation*}
    \begin{split}
      \int F \, \de {\tilde \beta}^N
      &
     = N\int_0^T \!\!\de t \int \Xi_t^N(d\eta) \, \\
      &\int q^{\eta\otimes \eta}(\de v,
  \de v_*,\de v',\de v'_*) F\bigl(t, \eta^{N,v,v_*,v',v'_*}, \eta, v', v'_*.
  v,v_*\bigr),
   \end{split}
\end{equation*}
}

From \eqref{AGS} we get
\begin{equation*}
\E\left(\mc V^{P^N}\vert \Upsilon_\#\mc V^{P^N}\right) {\carlo \ge \E\left(\beta^N \vert \tilde\beta^N \right)  = } N \E\Big(\frac{\beta^N}N\Big\vert \frac{\tilde\beta^N}N\Big),
\end{equation*}
therefore, by using \eqref{eq:EntN}
we obtain
\begin{equation}\label{ineqN}
\frac 1 N \Ent(P^N_t\vert \alpha^N)+\E\Big(\frac{\beta^N}N \Big\vert \frac{\tilde\beta^N}N\Big)\leq \frac 1 N \Ent(P^N_0\vert \alpha^N).
 \end{equation} 
We want to pass to the limit {\carlo as $N$ goes to infinity } in the above inequality. We will use two  Lemmata.

\begin{lemma}\label{lemma:convJ}
  For any  limit point $\Theta$ of $(\Theta^N)_{N\geq 1}$
  \begin{equation}\label{convJ}
    \liminf_{N \to +\infty} \E\Big(\frac{\beta^N}N \Big\vert \frac{\tilde\beta^N}N\Big)\geq \int \Theta(\de \eta, \de Q) \E(Q \vert \Upsilon_{\#}\ Q^{\eta\otimes\eta}).
  \end{equation}  
\end{lemma}  

\begin{proof}
  Denoting with $\Xi$ the first marginal of $\Theta$,
  for every $t \in [0,T]$ the family of measures $\Xi^N_t:=(e_t)_{\#}\Xi ^N$ weakly converges to $\Xi_t:=(e_t)_{\#}\Xi$.
Set 
\begin{align*}
  \beta(\de \eta, \de \zeta, \de v, \de v_*, \de v',\de v'_*)
  &=
    \de t \, \delta_{\eta}(\de \zeta) \, \Xi_t(\de \eta) \,
    q^{\eta \otimes \eta}(\de v, \de v_*,\de v',\de v'_*)\\
  \tilde   \beta(\de \eta, \de \zeta, \de v, \de v_*, \de v',\de v'_*)
  &=  \de t \, \delta_{\zeta}(\de \eta) \, \Xi_t(\de \eta) \,
  (\Upsilon_\#q^{\eta \otimes \eta})(\de v, \de v_*,\de v',\de v'_*).
\end{align*}
By Lemma
  \ref{lemma:lp}  $\Xi$ has support on $C([0,T],\ms P_e)$.
  Due to the linear growth of $B$, as
  $N$ diverges
  $\frac{\beta^N}N$ and $\frac{\tilde\beta^N}N$
  weakly converge to $\beta$
  and $\tilde \beta$ respectively.
  
By lower semicontinuity
\begin{equation*}
\liminf_{N \to +\infty} \E\Big(\frac{\beta^N}N \vert \frac{\tilde\beta^N}N\Big)\geq \E\big(\beta \vert \tilde \beta\big)
\end{equation*}  
so that $\beta$ is absolutely continuous with respect to $\tilde\beta$.

The Radon-Nikodym derivative   of $\beta$ with respect to $\tilde{\beta}$ is given by
\begin{equation*}
  \frac{\de \beta}{\de \tilde{\beta}}(t, \eta, \zeta, v, v_*, v',v'_*)
  = \frac{  \de q^{\eta \otimes \eta}\phantom{assa}}
  {\de (\Upsilon_\# q^{\eta \otimes \eta})} (v, v_*, v',v'_*)
   \qquad \tilde{\beta}\text{-a.e.}.
\end{equation*}
Hence 
\begin{align*}
  \E\big(\beta \vert \tilde \beta\big)
  &=
    \int_0^T \de t \int \Xi_t (\de \eta)\int \de q^{\eta\otimes \eta}\ln \frac{\de q^{\eta\otimes \eta}}{\de \big(\Upsilon_\# q^{\eta\otimes \eta}\big)} \\
  &=\int \Xi (\de \gamma) \int_0^T \de t 
    \int \de q^{\gamma_t\otimes \gamma_t}\ln \frac{\de q^{\gamma_t\otimes \gamma_t}}
    {\de \big(\Upsilon_\# q^{\gamma_t\otimes \gamma_t}\big)}\,,
\end{align*}
where in the last equality we used Fubini Theorem, due to the fact that  the argument
of the integral is non negative.

Thanks to Definition \ref{def:mfboltz} we have
$\de Q^{\gamma \otimes \gamma} = \de t  \de q^{\gamma_t \otimes \gamma_t} $ so that
$$\frac{\de q^{\gamma_t\otimes \gamma_t}}
{\de \big(\Upsilon_\# q^{\gamma_t\otimes \gamma_t}\big)} =  \frac{\de Q^{\gamma\otimes \gamma}}{\de \big(\Upsilon_\# Q^{\gamma\otimes \gamma}\big)}.$$
This implies
\begin{align*}
 \E\big(\beta \vert \tilde \beta\big) &=  \int \Xi(\de \gamma)\int \de Q^{\gamma\otimes \gamma}\ln 
\frac{\de Q^{\gamma\otimes \gamma}}{\de \big(\Upsilon_\# Q^{\gamma\otimes \gamma}\big)}.
\end{align*}
In the last integral the first marginal $\Xi$ can be replaced by the full probability measure $\Theta$. Using  that $Q^{\gamma\otimes\gamma}$ and $ \Upsilon_{\#}Q^{\gamma\otimes\gamma}$ have the same mass we obtain
\begin{equation*}
  \E\big(\beta \vert \tilde \beta\big)
  =\int \Theta(\de \gamma, \de Q)\,  \E(Q^{\gamma\otimes\gamma}\vert \Upsilon_{\#}Q^{\gamma\otimes\gamma}).
\end{equation*}
By Lemma \ref{lemma:lp}, $\Theta$ concentrates on pairs $(\gamma, Q)$ such that $Q=Q^{\gamma\otimes\gamma}$ a.s., therefore
\begin{equation*}
\int  \Theta(\de \gamma, \de Q)\, \E(Q^{\gamma\otimes\gamma}\vert \Upsilon_{\#} Q^{\gamma\otimes\gamma})=\int \Theta(\de \gamma, \de Q)\, \E(Q \vert \Upsilon_{\#} Q^{\gamma\otimes\gamma}).
\end{equation*}
\end{proof}

Recall the functional $\He$ defined in \eqref{def:He}.

\begin{lemma}\label{lemma:convent}
  For any    limit point $\Theta$ of $(\Theta^N)_{N\geq 1}$, for every $t\in [0,T]$
  \begin{equation}\label{convent}
    \liminf_{N\to \infty}
\frac 1 N\Ent\big(P^N_t\vert \alpha^N\big)\geq \int \Theta(\de \gamma, \de Q) \He(\gamma_t).
\end{equation}
\end{lemma}  
  
\begin{proof}
  Thanks to \eqref{AGS}, for any $t\in[0,T]$
\begin{equation*}
\Ent\big(P^N_t\vert \alpha^N\big)\geq \Ent\big (\Xi^N_t\vert (\pi^N)_\#\alpha^N\big).
\end{equation*}  
By Lemma \ref{lemma:lp}, for every $t \in [0,T]$, $\Xi^N_t:=(e_t)_{\#}\Xi^N$
weakly converges to $\Xi_t:=(e_t)_{\#}\Xi$ up to a subsequence.
Set $\mu^N:=(\pi^N)_\#\alpha^N$, $N\geq 2$.
In \cite{BBBC2}, Theorem 2.2. (already in  \cite{KR}) it has been proven that $(\mu^N)_N$ satisfies a large deviation principle with speed $N$ and rate function $H_e$.
Therefore, by using  Theorem 3.5 in \cite{Mar}, for every $t \in [0,T]$
\begin{equation*}
\liminf_N \frac 1 N \Ent(\Xi^N_t\vert \mu^N)\geq \int \Xi_t(\de \eta) \He(\eta) = \int \Xi(\de \gamma) \He(\gamma_t),
\end{equation*}  
from which \eqref{convent}
 follows.
 \end{proof}

 \begin{proof}[Proof of Theorem \ref{main_theo}]
Observe that, by definition,  $\Ent(\mu\vert M_e)\leq \He (\mu)$ for any probability measure $\mu$ on $\bb R^d$.
Since the initial datum is entropically chaotic
\begin{equation}
\lim_{N\to\infty} \frac 1 N \Ent(P^N_0\vert \alpha^N)=\Ent(P_0\vert M_e)
\end{equation}
By Lemma \ref{lemma:convent} with $t=0$
\begin{equation*}
\liminf \frac 1 N\Ent\big(P^N_0\vert \alpha^N\big)\geq \int \Xi_0(\de \eta) \He(\eta)=\He(P_0),
\end{equation*}  
therefore  $\Ent(P_0\vert M_e)\geq \He(P_0)$, which implies $\Ent(P_0\vert M_e)=\He(P_0)$, or,  equivalently $P_0(\zeta_0)=e$.

Passing to the limit in \eqref{ineqN} we obtain
\begin{equation*}
\int \Theta(\de \gamma, \de Q)\big\{ \He(\gamma_T) + \E(Q \vert \Upsilon_{\#} Q^{\gamma\otimes\gamma}) - \He(\gamma_0)\big\}\leq 0
\end{equation*}
Moreover, since  $Q=Q^{\gamma\otimes\gamma}$ $\Theta$-a.s., we can add a term $\E(Q \vert Q^{\gamma\otimes\gamma})$ in the integral and obtain
\begin{equation}
  \label{finale}
\int \Theta(\de \gamma, \de Q)\big\{ \He(\gamma_T) +\E(Q\vert Q^{\gamma\otimes\gamma})+ \E(Q \vert \Upsilon_{\#} Q^{\gamma\otimes\gamma}) - \He(\gamma_0)\big\}\leq 0.
\end{equation}
Recalling \eqref{heh},
in view of  Proposition \ref{prop:QQpi} and the  uniqueness of the variational solution,
stated in Theorem \ref{theo:eu},
we conclude that $\Theta=\delta_{(P,Q^{P\otimes P})}$, with $P_t(\zeta_0)=e$, for every $t\in[0,T]$.

It remains to show the convergence of the entropy.
This follows from the fact that on the unique variational solution the energy is conserved,
therefore $\He(P_T)-\He(P_0) = \Ent(P_T|M_e) - \Ent(P_0|M_e)$.
Using \eqref{finale} and that $\Theta=\delta_{(P,Q^{P\otimes P})}$
$$\Ent(P_T|M_e) + E(Q^{P\otimes P} | \Upsilon_\# Q^{P\otimes P}) = \Ent(P_0|M_e).$$
Taking into account \eqref{convJ} and
\eqref{convent}, by the entropic chaoticity, 
passing to the limit in \eqref{ineqN}, 
we conclude that
$$\lim_{N\to +\infty} \frac 1N \Ent(P^N_T|\alpha^N) =
\Ent(P_T,M_e).$$
\end{proof}

\section*{Acknowledgments}

The authors thank Lorenzo Bertini for suggesting the use of
variational techniques to derive entropic chaoticity.

Dario Benedetto  would like to thank GNFM - INdAM.

Carlo Orrieri  would like to thank GNAMPA - INdAM.

\section*{Declarations}

Giada Basile has been 
funded by the European Union (ERC CoG KiLiM, project number 101125162).
Views and opinions expressed are however those of the author(s)
only and do not necessarily reflect those of the European Union or the European Research Council. Neither the European Union nor the granting authority can be held responsible for them.

\end{document}